\documentclass[12pt,a4paper]{article}
\usepackage{url}
\usepackage[utf8]{inputenc}
\usepackage{tikz}
\usepackage{cancel}
\usepackage{color}
\usepackage[normalem]{ulem}
\usepackage[english]{babel}
\usepackage{color}
\usepackage[english]{babel}
\usepackage{color}
\usepackage{graphicx}
\usepackage{amsmath,amssymb,amsthm}
\usepackage[T1]{fontenc}
\usepackage[utf8]{inputenc}
\usepackage{authblk}
\usepackage[english]{babel}
\usepackage{amsmath,amssymb,amsthm}
\usepackage{color}
\usepackage{graphics}
\usepackage{amsmath,amssymb,amsthm}
\usepackage{lipsum}

\newtheorem{theorem}{Theorem}
\newtheorem{proposition}[theorem]{Proposition}
\newtheorem{corollary}[theorem]{Corollary}
\newtheorem{lemma}[theorem]{Lemma}
\theoremstyle{definition}
\newtheorem{definition}[theorem]{Definition}

\begin{document}
\title{On asymptotically good ramp secret sharing schemes
\footnote{Part of this paper was presented at the Ninth International
  Workshop on Coding and Cryptography (WCC 2015).}}

\author[1]{Olav Geil\thanks{olav@math.aau.dk}}
\author[1]{Stefano Martin\thanks{stefano@math.aau.dk}}
\author[1]{Umberto Mart\'inez-Pe\~nas\thanks{umberto@math.aau.dk}}
\author[1,2]{Ryutaroh Matsumoto\thanks{ryutaroh@it.ce.titech.ac.jp}}
\author[1]{Diego Ruano\thanks{diego@math.aau.dk}}
\affil[1]{Department of Mathematical
  Sciences, Aalborg University}
\affil[2]{Department of Information and Communications Engineering, Tokyo Institute of Technology}

\maketitle

\begin{abstract}
Asymptotically good sequences of linear ramp secret sharing schemes have been
intensively studied by Cramer et al.\ in terms of sequences of pairs of nested algebraic geometric codes~\cite{c2,MR3225937,MR2422182,Chen,Cramer-Cascudo,cramer3}. 
In those works the
focus is on full privacy and full reconstruction. In this paper we
analyze additional parameters describing the asymptotic behavior of partial information leakage
and possibly also partial reconstruction giving a more complete
picture of the access structure for sequences of linear ramp secret
sharing schemes. Our study involves a detailed treatment of the (relative)
generalized Hamming weights of the considered codes.\\

\noindent {\bf{Keywords:}} Algebraic geometric codes, generalized Hamming weights, relative generalized Hamming weights, secret sharing.
\end{abstract}

\section{Introduction}\label{sec1}
A secret sharing scheme \cite{shamir,blakleyoprindelig,blakley,yamamoto} is a cryptographic
method to encode a secret $ \mathbf{s} $ into multiple shares $ c_1,
 \ldots , c_n $ so that only from specified subsets of the shares
one can recover $ \mathbf{s} $. Often it is assumed that $ n $
participants each receive a share, no two different participants
receiving the same. The secret and the share vector $\mathbf{c}=(c_1,
 \ldots , c_n) $ corresponding to it are assumed to be taken at
random with some given distributions (usually uniform), and the
recovery capability of a set of shares is measured from an
information-theoretical point of view \cite{yamamoto}.  The term ramp secret sharing
scheme \cite{yamamoto,blakley,Chen} is used for those schemes where some
sets of shares partially determine the secret, but not
completely. This allows the shares to be of smaller size than the
secret. 

In this paper, we concentrate on linear ramp secret sharing schemes
with uniform distribution on the secret and uniform distribution on
the share vector conditioned to the secret, which is widely considered
in the literature (see, for instance,
\cite{MR2422182,Chen,rghw2014,kurihara}). Here,
the secret is a vector $ \mathbf{s} \in \mathbb{F}_q^\ell $ (for some
finite field $ \mathbb{F}_q $), and we assume that the shares are
elements $ c_1,  \ldots, c_n \in \mathbb{F}_q $. The term linear means that a linear combination of share vectors is a
share vector of the corresponding linear combination of
secrets. In~\cite[Sec.\ 4.2]{Chen} it was shown that such schemes
are equivalent to the following construction based on two nested
linear codes $C_2 \subsetneq C_1 \subset {\mathbb{F}}_q^n $ with $\dim
C_1-\dim C_2=\ell$. Writing $k_2 =\dim C_2$ and $k_1=\dim C_1$ (and
consequently $\ell=k_1-k_2$) let $\{
\mathbf{b}_1, \ldots , \mathbf{b}_{k_2}\}$ be a basis for $C_2$ and
extend it to a basis $\{\mathbf{b}_1, \ldots , \mathbf{b}_{k_1}\}$ for
$C_1$. A secret $\mathbf{s}=(s_1, \ldots , s_\ell)$ is encoded by
first choosing at random coefficients $a_1, \ldots , a_{k_2} \in
{\mathbb{F}}_q$ and then letting the share vector be
\begin{equation}
\mathbf{c}=a_1\mathbf{b}_1 + \cdots +
a_{k_2}\mathbf{b}_{k_2}+s_1\mathbf{b}_{k_2+1}+\cdots +s_\ell \mathbf{b}_{k_1}.
\label{eqm1}
\end{equation}
Define a $ q $-bit of information to be $ \log_2(q) $ bits of information. Then, for the schemes that we consider, the mutual information between the secret and a set of shares is an integer between $ 0 $ and $ \ell $ if measured in $ q $-bits \cite[Proof of Th.\ 4]{kurihara}. Therefore, for each $ m = 1, \ldots, \ell $, we may define the following threshold values \cite[Def.\ 2]{rghw2014}: 
\begin{itemize} 
\item The $ m $-th privacy threshold of the scheme is the maximum integer $
  t_m $ such that from no set of $ t_m $ shares one can recover $ m $
  $ q $-bits of information about the secret. That is, $t_m = \max\{
  \# J \mid J \subseteq \{1, \ldots , n\}, I(J) < m \}$,
    where $I(J) = 
    I\big( s_1$, \ldots, $s_\ell$ ;
    $( c_i \mid i \in J )\big)$. Here,  $c_i$ is the $i$-th component of
    $\mathbf{c}$ in~(\ref{eqm1}), and
  $I(;)$ is the mutual information taking logarithms in base $q$.
\item The $ m $-th reconstruction threshold of the scheme is the minimum
  integer $ r_m $ such that from any set of $ r_m $ shares one can obtain $ m $ $ q $-bits of information about $ \mathbf{s} $.
  That is, $r_m = \min\{ \# J \mid J \subseteq \{1, \ldots , n\}, I(J) \geq m \}$.
\end{itemize}
The numbers $ t = t_1 $ and $ r = r_\ell $ have been intensively
studied in the literature, e.g.\ \cite{blakley,Chen,yamamoto}, where they are called
privacy and reconstruction threshold,
respectively. Clearly $t$ is the greatest number such that no set of
$t$ shares holds any information on the secret and $r$ is the smallest
number such that from any set of $r$ shares one can reconstruct the
information in full. In a series of papers the asymptotic behavior of such
parameters has been investigated
\cite{c2,MR3225937,MR2422182,Chen,Cramer-Cascudo,cramer3} in terms
of corresponding infinite sequences of nested code pairs of increasing
length. In the present paper we take a
particular interest in sequences of nested code pairs
$(C_2(i) \subsetneq C_1(i) \subset {\mathbb{F}}_q^{n_i})_{i=1}^\infty$
with $n_i$ and with   
$\ell_i =\dim C_1(i)-\dim C_2(i)$ satisfying 
\begin{equation}
\lim_{i \rightarrow \infty}n_i = \infty ,{\mbox{ \
    and \  }} \liminf_{i
  \rightarrow \infty} ( \ell_i
/n_i)=L\label{eqLfixed}
\end{equation}
for some fixed $0 < L < 1$,
see~\cite{c2,MR3225937,MR2422182,Chen,Cramer-Cascudo,cramer3}. The
reason for us to require~(\ref{eqLfixed}) is to obtain a constant
information rate. For instance if the schemes are to be used in
connection with distributed storage as mentioned in~\cite{yamamoto}
then a memory of size $1/L$ times the information size is enough. As in
the above listed papers the focus in on full privacy and full
reconstruction, what is studied there is
\begin{equation}
\liminf_{i \rightarrow \infty }\frac{t}{n_i}=\Omega^{(1)} {\mbox { \ and \
  }} \limsup_{i \rightarrow \infty}\frac{r}{n_i}=\Omega^{(2)}
.\label{eqhuskmig}
\end{equation}
Here, $t$ and $r$ are the privacy and reconstruction thresholds for
the schemes based on $C_2(i) \subsetneq C_1(i) \subset
{\mathbb{F}}_q^{n_i}$, and thereby are functions in $i$. For any
chosen value of $L$ and corresponding feasible $\Omega^{(1)}$ it is desirable to have
the threshold gap $\Omega^{(2)}-\Omega^{(1)}$ as small as
possible. One way of achieving
this \cite{c2,MR3225937,MR2422182,Chen,Cramer-Cascudo,cramer3} is to
base the secret sharing schemes on sequences of nested code pairs
related to an optimal tower of function fields and to require $\lim_{i
  \rightarrow \infty}(\dim C_1(i)/n_i)=R_1$ and $\lim_{i
  \rightarrow \infty}(\dim C_2(i)/n_i)=R_2$ for some fixed rates $R_1 > R_2$. Using the Goppa
bound~\cite{handbook} one then obtains good parameters $L=R_1-R_2$, $\Omega^{(2)}$ and
$\Omega^{(1)}$. For future reference we formalize the concept of
asymptotic goodness in a
definition, where for completeness we also include the case $L=0$,
although we do not study this case in the present paper.
\begin{definition}\label{defas}
Let $0 < R_2 \leq R_1 < 1$ and consider a sequence of nested codes
$(C_2(i) \subsetneq C_1(i)\subset {\mathbb{F}}_q^{n_i})_{i=1}^\infty$
with $n_i \rightarrow \infty$, $\dim C_2(i)/n_i \rightarrow R_2$ and
$\dim C_1(i)/n_i \rightarrow R_1$ for $i \rightarrow \infty$. The
corresponding sequence of linear ramp secret sharing schemes is said
to be asymptotically good if the parameters from~(\ref{eqhuskmig})
satisfy $0 < \Omega^{(1)}$ and $\Omega^{(2)} < 1$.
\end{definition}

The purpose of the present paper is to
provide additional information on the access structure of
sequences of linear ramp secret sharing schemes by studying
partial information leakage and partial reconstruction parameters. More precisely, given a sequence of linear ramp secret
sharing schemes and any fixed numbers $0 \leq
\varepsilon_1, \varepsilon_2 \leq 1$ we study the asymptotic parameters
\begin{eqnarray}
\Lambda^{(1)}(\varepsilon_1)&=&\sup \big\{ \liminf_{i
  \rightarrow \infty}\frac{t_{m_1(i)}}{n_i} \mid (m_1(i))_{i=1}^\infty
{\mbox { satisfies }} \nonumber \\
&&{\mbox{ \hspace{0.5cm} }}1 \leq m_1(i) \leq \ell_i, \lim_{i \rightarrow
  \infty} (m_1(i)/n_i)=\varepsilon_1L \big\}, \nonumber \\
\Lambda^{(2)}(\varepsilon_2)&=&\inf \big\{ \limsup_{i
  \rightarrow \infty}\frac{r_{\ell_i -m_2(i)+1}}{n_i} \mid (m_2(i))_{i=1}^\infty
{\mbox { satisfies }} \nonumber \\
&&{\mbox{ \hspace{0.5cm} }}1 \leq m_2(i) \leq \ell_i, \lim_{i \rightarrow
  \infty} (m_2(i)/n_i)=\varepsilon_2L \big\} .\nonumber
\end{eqnarray}
Such
parameters tell us that asymptotically no fraction less than
$\Lambda^{(1)}(\varepsilon_1)$ 
of the shares holds more information on the secret than a fraction $\varepsilon_1$. Similarly, from any fraction greater than
$\Lambda^{(2)}(\varepsilon_2)$ of the shares one can gain information on
the secret corresponding to a fraction $1-\varepsilon_2$ or more. Of particular interest is $\Lambda^{(1)}(0)$ which ensures almost
full privacy. It is a surprising fact that for secret sharing schemes
based on algebraic geometric codes this number can be
significantly larger than $\Omega^{(1)}$, meaning that such schemes
are more secure than anticipated (see Section~\ref{secnewas} and Theorem~\ref{final2}). The situation is similar
with regards to reconstruction. In another direction, for fixed values
of $L$ and corresponding feasible $\Lambda^{(1)}(\varepsilon_1)$ we determine for the general
class of ramp secret sharing schemes the
smallest value $\Lambda^{(2)}(\varepsilon_2)$ such that a sequence of
codes with these parameters exists. This bound -- which
can be seen as an asymptotic Singleton bound for linear ramp secret sharing
schemes -- is then by a non-constructive proof shown to be
achievable, but unfortunately, we obtain no information regarding
$\Omega^{(1)}$ and $\Omega^{(2)}$ for those sequences. 

Sequences
of linear ramp secret sharing schemes based on algebraic geometric
codes defined from optimal towers of function
fields are interesting for the following three reasons. Firstly, for such sequences the parameters
$L$, 
$\Omega^{(1)}$ and   $\Omega^{(2)}$ are simultaneously good. Also
$\Lambda^{(1)}(\varepsilon_1)$ and $\Lambda^{(2)}(\varepsilon_2)$ 
are good, although they do not always reach the Singleton bound. Secondly,
such sequences are constructible if $q$ is a perfect square and are
semi-constructible if not. Finally, as demonstrated in~\cite{c2,MR3225937,MR2422182,Chen,Cramer-Cascudo,cramer3}
examples of such sequences are important in connection with secure
multiparty computation due to nice properties on the componentwise
product of share vectors.

Our analysis of the asymptotic secret sharing parameters is based on
the material in \cite{rghw2014,kurihara} which
translates information-theoretical properties of a ramp secret sharing
scheme based on nested linear codes $C_2 \subsetneq C_1 \subset
{\mathbb{F}}_q^n$ into coding-theoretical properties of the nested
codes. In particular, bounding generalized Hamming weights \cite{wei} of $ C_1 $ and $ C_2^\perp $ and relative generalized
Hamming weights \cite{luo} of the pairs $ C_2
\varsubsetneq C_1 $ and $ C_1^\perp \varsubsetneq C_2^\perp $ implies
bounds on the privacy and reconstruction numbers $t_i$ and $r_i$.

The paper is organized as follows. In Section \ref{sec3} we give
the Singleton bound for linear ramp secret sharing schemes. Using
the material from~\ref{app1} we then 
show that for arbitrary $L$,  
sequences of schemes exist such that for arbitrary $\varepsilon_1, \varepsilon_2$ 
one
gets arbitrarily close to the Singleton bound for $\Lambda^{(1)}(\varepsilon_1)$ and $\Lambda^{(2)}(\varepsilon_2)$. In
Section~\ref{secnewas} we then discuss how to obtain sequences of ramp secret
sharing schemes with good values of $L$, $\Omega^{(1)}$ and $\Omega^{(2)}$
from optimal towers of function fields.  
As a preparation step to treat later in the paper
$\Lambda^{(1)}(\varepsilon_1)$ and $\Lambda^{(2)}(\varepsilon_2)$ for
these sequences of schemes we next study relative generalized Hamming weights of
algebraic geometric codes in
Section~\ref{sec4} and derive asymptotic
consequences in Section~\ref{sec5}. Then finally in Section~\ref{sec7}
we collect our findings into information on $\Omega^{(1)}$, $\Omega^{(2)}$, $\Lambda^{(1)}(\varepsilon_1)$ and
$\Lambda^{(2)}(\varepsilon_2)$ for sequences of ramp secret sharing
schemes based on algebraic geometric codes coming from optimal towers
of function fields.

\section{The Singleton bound}\label{sec3}
The code parameters governing the privacy and reconstruction
numbers $t_m$ and $r_m$ of linear ramp secret sharing schemes are the relative generalized Hamming weights~\cite{luo} which we now
define together with the generalized Hamming weights~\cite{wei}.

\begin{definition}
Consider $C_2 \subsetneq C_1 \subset {\mathbb{F}}_q^n$ and let $\ell
=k_1-k_2$ where $k_1 =\dim C_1$ and $k_2=\dim C_2$. For $m=1, \ldots ,
\ell$ the $m$-th relative generalized Hamming weight (RGHW) is:
\begin{eqnarray}
{\mbox{\hspace{-0.5cm}}}M_m(C_1,C_2) &=& \min \{  \# {\rm Supp}(D) \mid D \subset C_1 \textrm{
                 is a linear space} \nonumber  \\
 && {\mbox{ \ \ \ \ \ \  with }}\dim(D) = m \textrm{ and } D \cap C_2 = \{ \mathbf{0} \} \},\nonumber
\end{eqnarray}
where $ {\rm Supp}(D) = \{ i \in \{ 1,2, \ldots, n \} \mid \exists
\mathbf{d} \in D, d_i \neq 0 \} $. For $ m = 1,2, \ldots, k_1 $, the $
m $-th generalized Hamming weight (GHW) of $ C_1 $ is defined as $ d_m(C_1) = M_m(C_1, \{ \mathbf{0} \}) $.
\end{definition}

Clearly, the RGHWs can be lower bounded by the GHWs of the same
index, and as the latter are often easier to estimate we shall also take
an interest in them. The following theorem, which is \cite[Th.\ 3]{rghw2014}, gives a characterization of the threshold numbers $ t_m $ and $ r_m $ in terms of the RGHWs of the pairs $ C_2 \varsubsetneq C_1 $ and $ C_1^\perp \varsubsetneq C_2^\perp $, where $ C^\perp $ denotes the dual of the linear code $ C $.

\begin{theorem} \label{thresholds and RGHW}
Consider a linear ramp secret sharing scheme based on codes $C_2
\subsetneq C_1 \subset {\mathbb{F}}_q^n$. Then for $ m = 1,2, \ldots, \ell $,
\begin{eqnarray}
t_m&=&M_m(C_2^\perp, C_1^\perp)-1 \nonumber, \textrm{ and} \\
r_m&=&n-M_{\ell-m+1}(C_1,C_2)+1. \nonumber
\end{eqnarray}
\end{theorem}

Observe, that as a consequence we obtain $t_m \geq d(C_2^\perp)-1$ and
$r_m \leq n- d_{\ell-m +1}(C_1)+1$. Given a sequence of linear ramp
secret sharing schemes satisfying~(\ref{eqLfixed}), numbers $0 \leq
\varepsilon_1, \varepsilon_2 \leq 1$ and any two sequences
$(m_1(i))_{i=1}^\infty$ and $(m_2(i))_{i=1}^\infty$ with $\lim_{i
  \rightarrow \infty} (m_1(i)/n_i) \rightarrow \varepsilon_1 L$ and $\lim_{i
  \rightarrow \infty} (m_2(i)/n_i) \rightarrow \varepsilon_2 L$ we then 
obtain
\begin{eqnarray}
\Omega^{(1)}&=&\liminf_{i\rightarrow \infty}\frac{M_1(C_2^\perp ,
                C_1^\perp)}{n_i} \geq \liminf_{i \rightarrow \infty}
                \frac{d(C_2^\perp)}{n_i} \label{eqomega1bnd}\\
\Omega^{(2)}&=&1-\liminf_{i\rightarrow \infty}\frac{M_1(C_1 ,
                C_2)}{n_i} \nonumber \\
& \leq & 1- \liminf_{i \rightarrow \infty}
                \frac{d(C_1)}{n_i} \label{eqomega2bnd}\\
\Lambda^{(1)}(\varepsilon_1)&\geq&\liminf_{i \rightarrow \infty}
                                \frac{M_{m_1(i)}(C_2^\perp,C_1^\perp)}{n_i}
  \label{eqlambda1bnd1} \\
                                &\geq &\liminf_{i \rightarrow \infty}
                                \frac{d_{m_1(i)}(C_2^\perp)}{n_i}
                                \label{eqlambda1bnd2} \\
\Lambda^{(2)}(\varepsilon_2)&\leq&1-\liminf_{i \rightarrow \infty}
                                \frac{M_{m_2(i)}(C_1,C_2)}{n_i}
                                \label{eqlambda2bnd1} \\
&\leq& 1- \liminf_{i \rightarrow \infty}
                                \frac{d_{m_2(i)}(C_1)}{n_i}
                                \label{eqlambda2bnd2}
\end{eqnarray}
To study the optimality of linear ramp secret sharing schemes we
recall the Singleton bound \cite[Section IV]{luo} for a linear code pair $ C_2 \varsubsetneq C_1 \subset \mathbb{F}_q^n $ and its dual pair $ C_1^\perp \varsubsetneq C_2^\perp \subset \mathbb{F}_q^n $: for each $ m = 1,2, \ldots, \ell $,
\begin{equation}
M_m(C_1,C_2) \leq n - k_1 + m, \quad \textrm{and} \quad M_m(C_2^\perp,C_1^\perp) \leq k_2 + m.
\label{singleton bound}
\end{equation}
From these bounds and Theorem \ref{thresholds and RGHW}, it follows
that $r_m \geq k_2 + m$ and  $t_m \leq k_2 + m -1$,
and as a consequence
\begin{equation}
\Omega^{(2)}-\Omega^{(1)} \geq L \label{eqdiffgam}
\end{equation}
and 
\begin{equation}
\Lambda^{(2)}(\varepsilon_2)-\Lambda^{(1)}(\varepsilon_1) \geq L(1-\varepsilon_1-\varepsilon_2).\label{eqdifflam}
\end{equation}
There exist choices of $\Omega^{(1)}<\Omega^{(2)}$ such
that~(\ref{eqdiffgam}) is not nearly tight, meaning that $L$ cannot be
close to $\Omega^{(2)}-\Omega^{(1)}$~\cite[Th.\ 3.26, Th.\ 4.6]{c2}. It is therefore surprising that for any fixed value of
$\Lambda^{(1)}(0) < \Lambda^{(2)}(0)$ there exist sequences of linear ramp secret
sharing schemes with $L$ arbitrarily close to $\Lambda^{(2)}(0)-\Lambda^{(1)}(0)$. Even more, by the
strict monotonicity of RGHWs~\cite[Pro.\ 2]{luo}, for such schemes $L(1-\varepsilon_1-\varepsilon_2)$ becomes arbitrarily
close to $\Lambda^{(2)}(\varepsilon_2)-\Lambda^{(1)}(\varepsilon_1)$
for all $0 \leq \epsilon_1, \epsilon_2 \leq 1$. Our proof is non-constructive, as might be expected, and it
unfortunately does not reveal any non-trivial information on the corresponding
values of $\Omega^{(1)}$ and $\Omega^{(2)}$. We leave it for further
research to determine simultaneous information on these parameters,
and in particular to decide if the sequences fulfill the requirements
in Definition~\ref{defas} for being asymptotically good. In~\ref{app1} we prove
the following
result:

\begin{theorem} \label{th extended}
For $0 \leq R_2 < R_1 \leq 1$, $0 \leq \delta \leq 1$, $0 \leq
\delta^\perp \leq 1$, $0 < \tau \leq
\min \{ \delta, R_1-R_2\}$ and $0 < \tau^\perp \leq \min \{
\delta^\perp, R_1-R_2 \}$, if 
\begin{eqnarray}
R_1+\delta < 1+\tau {\mbox{ \  and \  }} (1-R_2)+\delta^\perp < 1+ \tau^\perp, \label{eqnummer8}
\end{eqnarray}
then for any prime power $q$ there exists an infinite sequence of
nested linear code pairs $C_2(i) \subsetneq C_1(i) \subset
{\mathbb{F}}_q^{n_i}$, where $n_i \rightarrow \infty$ for $i
\rightarrow \infty$, and where
$$\lim_{i \rightarrow \infty}\frac{\dim (C_1(i))}{n_i}=R_1, $$
$$\lim_{i \rightarrow \infty}\frac{\dim (C_2(i))}{n_i}=R_2, $$
$$\liminf_{i \rightarrow \infty}\frac{M_{\lceil n_i \tau \rceil}(C_1(i),C_2(i))}{n_i}\geq\delta, \textrm{ and} $$ 
$$\liminf_{i \rightarrow \infty}\frac{M_{\lceil n_i \tau^\perp
    \rceil}(C_2(i)^\perp,C_1(i)^\perp)}{n_i}\geq  \delta^\perp.$$ 
\end{theorem}

As a corollary we see that the difference in (\ref{eqdifflam})
  can become arbitrarily close to zero.

\begin{corollary}
For any $0 < R_2 < R_1 < 1$ there
exists a sequence of linear ramp secret sharing schemes
satisfying~(\ref{eqLfixed}) with $L=R_1-R_2$ and having
simultaneous $\Lambda^{(1)}(\varepsilon_1)$ arbitrarily close to $R_2+\varepsilon_1 L$
and $\Lambda^{(2)}$ arbitrarily close to $R_1-\varepsilon_2L$ for all
$0\leq \varepsilon_1, \varepsilon_2 \leq 1$. 
\end{corollary}
\begin{proof}
As noted prior to Theorem~\ref{th extended} by the strict
monotonicity of the RGHWs it is enough to prove
$L=R_1-R_2$ and that  $\Lambda^{(1)}(0)$ can be arbitrarily close
to $R_2$ simultaneously with $\Lambda^{(2)}(0)$ being arbitrarily
close to $R_1$. We start
by proving a result which at a first glance seems weaker -- but from
which the above will follow. Let $0 < \varepsilon \leq \min \{
R_1/L,(1-R_2)/L\}$ and choose arbitrarily small $\mu >0$. In
Theorem~\ref{th extended} choose $\tau=\tau^\perp=\varepsilon L$,
$\delta=1-R_1+\varepsilon L-\mu$ and $\delta^\perp =R_2+\varepsilon L-\mu$.
By inspection all the conditions of the theorem
are satisfied and therefore by~(\ref{eqlambda1bnd1}) and
(\ref{eqlambda2bnd1}) for any $\varepsilon$
in the considered interval there exists a sequence of linear ramp secret sharing schemes
satisfying~(\ref{eqLfixed}) such that $\Lambda^{(1)}(\varepsilon)$ is arbitrarily
close to $R_2+\varepsilon L$ simultaneously with
$\Lambda^{(2)}(\varepsilon)$ being arbitrarily close to $R_1-\varepsilon
L$. The theorem finally follows by considering a sequence of numbers
$(\varepsilon(i))_{i=1}^\infty$  between $0$ and $\min \{R_1/L,(1-R_2)/L\}$ and with $\lim_{i
  \rightarrow \infty} \varepsilon(i)=0$. For each $\varepsilon(i)$ we have a
sequence ${\mathcal{S}}(i)$ of secret sharing schemes as described above. Now build a new
sequence of schemes in which the $i$-th scheme is the $i$-th scheme
from the sequence ${\mathcal{S}}(i)$. The resulting scheme satisfies the requirement
mentioned at the beginning of the proof.
\end{proof}
\section{Asymptotically good sequences of  schemes from algebraic
  geometric codes}\label{secnewas}

In the remaining part of the paper we concentrate on ramp secret sharing schemes
defined from pairs of nested algebraic geometric codes. In
the present section we collect known information to describe what is
possible concerning the parameters $L$, $\Omega^{(1)}$ and
$\Omega^{(2)}$. In subsequent sections we then derive information on $\Lambda^{(1)}(\varepsilon_1)$ and
$\Lambda^{(2)}(\varepsilon_2)$. 

Let $ \mathcal{F} $ be an algebraic function field over $ \mathbb{F}_q
$ of transcendence degree one. In the rest of the paper we consider divisors $ D = P_1  + \cdots +
P_n $ and $ G $ with disjoint supports, where the places $ P_i $ are
rational and pairwise distinct. For any divisor $ E $, we define the
Riemann-Roch space $ \mathcal{L}(E) $ of functions $ f \in \mathcal{F}
$ such that the divisor $ (f) + E $ is effective (see also
\cite[Def.\ 2.36]{handbook}). We denote by $ C_\mathcal{L}(D,G) $
the evaluation code of length $ n $ obtained by evaluating functions $
f \in \mathcal{L}(G) $ in the places $ P_i $. An algebraic geometric
code  is a code of the form $ C_\mathcal{L}(D,G) $ or $
C_\mathcal{L}(D,G)^\perp $. We call the first primary algebraic
geometric codes and the latter dual. The well-known Goppa bound~\cite[Th.\ 
2.65]{handbook} gives information on the relation between dimension
and minimum distance for primary or dual codes.

\begin{theorem} \label{thegoppa}
Let $ C $ be an algebraic geometric code of dimension $ k $ defined
from a function field of genus $g$. Then the minimum distance
satisfies $ d(C) \geq n - k + 1 - g $.
\end{theorem}

Given a function field ${\mathcal{F}}$, we shall write $N(\mathcal{F})$ for its number of rational places and $ g(\mathcal{F}) $ for its genus. For asymptotic purposes, we will make use of Ihara's constant \cite{ihara}
\begin{equation*}
A(q) = \limsup_{g(\mathcal{F}) \rightarrow \infty} \frac{N(\mathcal{F})}{g(\mathcal{F})},
\end{equation*}
where the limit is taken over all function fields over $ {\mathbb{F}}_q $ of genus $ g(\mathcal{F}) > 0 $. The Drinfeld-Vl\u{a}du\c{t} bound \cite{drvl} states that 
\begin{equation}
A(q) \leq \sqrt{q}-1. 
\label{eqsnabel}
\end{equation}
As is well-known $A(q)$ is always strictly positive and equality in~(\ref{eqsnabel}) holds if $ q $ is a perfect square \cite{ihara}. See \cite{bassabeelenmv} for the status on what is known about $ A(q) $ for $ q $ being a non-square. For convenience, we give the following definition:
\begin{definition} \label{optimal tower}
A tower of function fields $({\mathcal{F}}_i)_{i=1}^\infty$ over
${\mathbb{F}}_q$ is optimal if $N({\mathcal{F}}_i)\rightarrow \infty$
and $N({\mathcal{F}}_i)/g({\mathcal{F}}_i) \rightarrow A(q)$ for $i
\rightarrow \infty$. On the other hand, $(C_i)_{i=1}^\infty$ is an optimal sequence of one-point algebraic geometric codes defined from ${\mathcal{F}}_i$ if $n_i/N({\mathcal{F}}_i)\rightarrow 1$ for $i \rightarrow \infty$, where $ n_i $ is the length of $ C_i $.
\end{definition}

The above together with~(\ref{eqomega1bnd}) and (\ref{eqomega2bnd}) immediately combine into the following result concerning the existence
of 
asymptotically good sequences of ramp secret sharing schemes.

\begin{theorem}\label{theomega}
Let $(C_2(i) \subsetneq C_1(i)\subset
{\mathbb{F}}_q^{n_i})_{i=1}^{\infty}$ be a sequence of nested
algebraic geometric codes defined from an optimal tower of function
fields and satisfying $n_i=N({\mathcal{F}}_i)-1$, $\dim C_1(i)/n_i
\rightarrow R_1$ and $\dim C_2(i)/n_i \rightarrow R_2$ for some $0 <
R_2 \leq  R_1 < 1$. Then the corresponding sequence of linear ramp
secret sharing schemes (see Section~\ref{sec1}) satisfies 
$\Omega^{(1)} \geq R_2-\frac{1}{A(q)}$ and $\Omega^{(2)}
\leq  R_1+\frac{1}{A(q)}$. 
\end{theorem}

In particular we obtain asymptotically good ramp secret sharing
schemes (Definition~\ref{defas}) if $\frac{1}{A(q)} < R_2 \leq R_1 < 1-
\frac{1}{A(q)}$. If moreover $R_2 < R_1$ then also the crucial 
requirement~(\ref{eqLfixed}) is satisfied. Observe that due to the assumption $n_i=N({\mathcal{F}}_i)-1$ we may choose the codes in
Theorem~\ref{theomega} as one-point codes, meaning that without loss of generality we may consider codes of the form $C_2(i)=C_{\mathcal{L}}(D,\mu_2(i)Q)$ and
$C_1(i)=C_{\mathcal{L}}(D,\mu_1(i)Q)$, where $D$ is the sum of $n_i$
distinct rational places in ${\mathcal{F}}_i$ and $Q$ is another rational
place in the same function field.
\section{RGHWs and GHWs of algebraic geometric codes} \label{sec4}

In this section, we give non-asymptotic analysis that are necessary
in Sections \ref{sec5} and \ref{sec7} to treat the parameters
$\Lambda^{(1)}(\varepsilon_1)$ and $\Lambda^{(2)}(\varepsilon_2)$ of the
sequences of algebraic geometric schemes discussed in the previous section. The next theorem combines \cite[Th.\ 
2.65]{handbook}, \cite[Th.\  4.3, Cor.\ 4.2]{tsfasman95} and \cite[Th.\ 1]{wei}. The first part which is a generalization
of Theorem~\ref{thegoppa} is known as the Goppa
bound for GHWs.

\begin{theorem} \label{goppa bound}
Let $ C $ be an algebraic geometric code of dimension $ k $ defined from a function field of genus $g$. Then $ d_m(C) \geq n - k + m - g $, for $ 1 \leq m \leq g $, and $ d_m(C) = n - k + m $, for $ g + 1 \leq m \leq k $.
\end{theorem}

For algebraic geometric codes $C_2 \varsubsetneq C_1$,  the above
theorem exactly gives $d_m(C_1)$ and $M_m(C_1,C_2)$ when $g < m$. In
Proposition~\ref{proposition1} and Proposition~\ref{proposition2}
below, we will improve it in the case $ m \leq g $ for one-point
codes. From now on we will concentrate on one-point algebraic
geometric codes. That is,  codes $ C_{\mathcal{L}}(D,G) $ or $ C_{\mathcal{L}}(D,G)^\perp $, where $ G = \mu Q $, $Q$ is a rational place and $ \mu \geq -1 $. Writing  $ \nu_Q $ for the valuation at $ Q $, the Weierstrass semigroup corresponding to $ Q $ is
$$ H(Q) = -\nu_Q \left( \bigcup_{\mu = 0}^\infty \mathcal{L}(\mu Q) \right) = \{ \mu \in \mathbb{N}_0 \mid \mathcal{L}(\mu Q) \neq \mathcal{L}((\mu -1) Q) \}. $$
As is well-known, the number of missing positive numbers in $ H(Q) $
equals the genus $ g $ of the function field. The conductor $ c $ is
by definition the smallest element in $ H(Q) $ such that all integers
greater than or equal to that number belong to the set. The following lemma is well-known \cite[Th.\ 2.65]{handbook}:
\begin{lemma} \label{lemma 1}
For $ \mu \geq - 1 $, $ k = \dim C_{\mathcal{L}}(D,\mu Q)$ satisfies:
\begin{itemize}
 \item  $ k \geq \mu+1-g$ if $\mu \leq 2g-2$,
 \item $ k = \mu+1-g$ if $2g-2 < \mu <n$, and
\item $ k \leq \mu+1-g$ if $n \leq \mu$.   
\end{itemize}
If $ \mu = n+2g-1 $, then $ C_{\mathcal{L}}(D,\mu Q) = \mathbb{F}_q^n
$.
\end{lemma}
From~\cite[Th.\ 19, 20]{rghw2014} we have the following result.
\begin{theorem} \label{theorem 1} Let
$ C_1 = C_{\mathcal{L}}(D,\mu_1 Q) $ and $ C_2 = C_{\mathcal{L}}(D,\mu_2 Q) $, with $ -1 \leq \mu_2 < \mu_1 $. Write $ k_1 = \dim C_1 $, $ k_2 = \dim C_2 $ and $ \ell = k_1 - k_2 $. If $ 1 \leq m \leq \ell $,  then 
\begin{enumerate}
\item
$ M_m(C_1,C_2) \geq n - \mu_1 + \min \{ \# \{ \alpha \in
  \cup_{s=1}^{m-1}(i_s + H(Q)) \mid \alpha \notin H(Q) \} \mid -(\mu_1-\mu_2) +1 \leq i_1 < \ldots < i_{m-1} \leq -1 \} $.
\item
$ M_m(C_2^\perp,C_1^\perp) \geq \min \{ \# \{ \alpha \in \cup_{s=1}^{m}(i_s + (\mu_1 - H(Q))) \mid \alpha \in H(Q) \} \mid -(\mu_1-\mu_2)+1 \leq i_1 < \ldots < i_{m} \leq 0 \} $.
\end{enumerate}
\end{theorem}
Choosing $C_2 = \{ \mathbf{0} \}$ in item 1, we obtain a bound on the GHWs of $C_1$. Similarly, choosing $C_1 = \mathbb{F}_q^n$ in item 2, we get a bound on the GHWs of $C_2^\perp$. 
\begin{proposition} \label{proposition1}
For $ 0 \leq \gamma \leq c $, let $ h_\gamma = \# \left( H(Q) \cap (0,\gamma] \right) $ and let $ \mu \geq -1 $ and $k=\dim C_{\mathcal{L}}(D,\mu Q)$. If $ \mu < n $ and $ 1 \leq m \leq \min \{ k,g\} $, then 
  $$ d_m(C_{\mathcal{L}}(D,\mu Q)) \geq n-k + 2m - c + h_{c-m} \geq n-k+2m-c.$$
\end{proposition}
\begin{proof}
We will apply item 1 in Theorem \ref{theorem 1} for $ \mu_1 = \mu $ and $ \mu_2 = - 1 $. Consider numbers $ -\mu \leq i_1 <  \cdots < i_{m-1} \leq -1 $. We have  $ [c-m+1,c] \setminus H(Q) \subset [\max \{0,c+i_1 \},c] \setminus H(Q) \subset \{ \alpha \in \cup_{s=1}^{m-1}(i_s + H(Q)) \mid \alpha \notin H(Q) \} \cap [0,\infty) $, where the first inclusion comes from $ i_1 \leq -m+1 $. Now the number of elements in $ [c-m+1,c] \cap H(Q) $ is at most $ (c-g) - h_{c-m} $, and we have that
$ \# \left( \{ \alpha \in \cup_{s=1}^{m-1}(i_s + H(Q)) \mid \alpha \notin H(Q) \} \cap [0,\infty) \right) \geq m - (c-g) + h_{c-m}$.
On the other hand, we have that $ \{ i_1,  \ldots, i_{m-1} \}
\subset \{ \alpha \in \cup_{s=1}^{m-1}(i_s + H(Q)) \mid \alpha \notin
H(Q) \} \cap (-\infty,0) $. Thus, from Theorem~\ref{theorem 1}, we obtain  
$ d_m(C_{\mathcal{L}}(D,\mu Q)) \geq (n-\mu) + (m - 1) + (m - c + g + h_{c-m})$.
Since $ k \geq \mu - g + 1 $ by Lemma~\ref{lemma 1}, the result follows.
\end{proof}
\begin{proposition} \label{proposition2}
For $ \gamma \geq 1 $, let $ h^\prime_\gamma = \# ([\gamma,\infty)
\setminus H(Q)) $ and let $ \mu > 2g-2 $ and $k=\dim C_{\mathcal{L}}(D,\mu Q)^\perp $.
If $ 1 \leq m \leq \min \{ k,g \} $, then 
$$ d_m(C_{\mathcal{L}}(D,\mu Q)^\perp) \geq n-k + 2m -c + h^\prime_{\mu-c+m}\geq n-k+2m-c. $$
\end{proposition}
\begin{proof}
We will apply item 2 in Theorem \ref{theorem 1} for $ \mu_1 = n + 2g -
1 $ and $ \mu_2 = \mu $ to prove that $M_m(C_2^\perp,C_1^\perp) \geq
k_2+2m-c+h^\prime_{\mu_2-c+m}$, where $ k_2 = \dim C_2 $. 
Consider numbers $ -(\mu_1-\mu_2)+1 \leq i_1 <  \cdots < i_{m} \leq 0 $. First, $ (i_m+\mu_1 - H(Q)) \cap [0,\mu_2] $ contains the set $ [0,\mu_1-c-(\mu_1-\mu_2)+m] = [0,\mu_2 -c+m] $, since $ i_m \geq -(\mu_1-\mu_2)+m $ and $ \mu_1-c-(\mu_1-\mu_2)+m \leq \mu_2 $. Here, we used the assumption $m \leq g$ and the fact that $g \leq c$. Thus,
$ \# \left( (i_m+\mu_1 - H(Q)) \cap H(Q) \cap [0,\mu_2] \right)$ is
greater than or equal to $(\mu_2 -c + m + 1) - (g - h^\prime_{\mu_2 - c + m})$.
On the other hand, $ \{ \mu_1 + i_1,  \ldots, \mu_1 +
i_{m} \}$ is contained in $\{ \alpha \in \cup_{s=1}^{m}(i_s + (\mu_1 -
H(Q))) \mid \alpha \in H(Q) \}$, which are $ m $ elements in the range
$ (\mu_2,\mu_1] $. Thus, from the previous theorem we obtain $ M_m(C_2^\perp,C_1^\perp) \geq (\mu_2-c+m+1 - g + h_{\mu_2-c+m}) + m$.
Since $ k_2 \leq \mu_2-g+1 $ and $C_1={\mathbb{F}}_q^n$ by Lemma~\ref{lemma 1}, the result follows.
\end{proof}
%
\section{Asymptotic analysis for algebraic geometric codes} \label{sec5}
As a preparation step to treat the parameters
$\Lambda^{(1)}(\varepsilon_1)$ and $\Lambda^{(2)}(\varepsilon_2)$ of
sequences of schemes based on algebraic geometric codes, in this
section we derive asymptotic consequences of the non-asymptotic results
derived in the previous section. We start our investigations by
commenting on \cite[Th.\ 5.9]{tsfasman95}, which if true would imply
that the codes in Theorem~\ref{theomega} would attain the Singleton bound~(\ref{eqdifflam}) in all cases 
$\frac{1}{q} < R_2 < R_1 < 1-\frac{1}{q}$ and for all $0 \leq
\varepsilon_1 , 
\varepsilon_2 \leq 1$. Below we reformulate~\cite[Th.\ 5.9]{tsfasman95} with the needed modification which ensures that the Singleton bound is reached when $1/A(q) < \rho$, in contrast to $0 \leq \rho$, as it appears in~\cite{tsfasman95}. We also adapt the formulation to better fit our purposes of constructing asymptotically good sequences of secret sharing schemes. We include the proof from~\cite{tsfasman95} to explain why this modification is needed. 
\begin{theorem} \label{theorem tsfasman}
  Let $({\mathcal{F}}_i)_{i=1}^\infty$ be an optimal tower of function fields over ${\mathbb{F}}_q$. Consider $R, \rho$ with $ 0 \leq \rho \leq R \leq 1$.
  Let $(C_i)_{i=1}^\infty$ be an optimal sequence of one-point algebraic geometric codes defined from $({\mathcal{F}}_i)_{i=1}^\infty$
  such that $\dim C_i/n_i \rightarrow R$. For all sequences of positive integers $(m_i)_{i=1}^\infty$ with $m_i/n_i \rightarrow \rho$, it holds that
$ \delta = \liminf_{i \rightarrow \infty} d_{m_i}(C_i)/n_i \geq 1-R+\rho-\frac{1}{A(q)} $ and, if $ 1/A(q) < \rho$, then $\delta=1-R+\rho$.
\end{theorem}
\begin{proof}
The first bound on $\delta$ is an easy consequence of the Goppa bound (the first part of Theorem~\ref{goppa bound}). Now assume $1/A(q) < \rho$. 
By assumption, for $i$ large enough we have $ m_i >
g({\mathcal{F}}_i)$, which by the last part of Theorem \ref{goppa
  bound} implies that $ d_{m_i}(C_i) = n_i - \dim C_i + m_i $. Dividing by $ n_i $ and taking the limit, we obtain the result.
\end{proof}

The theorem states that the Singleton bound (\ref{singleton bound})
can be asymptotically reached when $ 1/A(q) < \rho $, which implies $
1/(\sqrt{q}-1) < \rho $ by (\ref{eqsnabel}). However, this leaves the
cases $ 1/A(q) \geq \rho $ undecided. In the following, we shall
concentrate on finding asymptotic results for the cases $ 1/A(q) \geq
\rho $. We will need \cite[Cor.\ 3.6]{tsfasman95} and Wei's duality theorem \cite[Th.\ 3]{wei}, which we now recall in this order:
\begin{lemma}\label{cor2}
For every linear code $ C \subset \mathbb{F}_q^n $ 
we have that
$$ d_m(C) \geq d_1(C) \frac{q^m-1}{q^m-q^{m-1}}, \quad   m= 1, \ldots , \dim
C.$$
\end{lemma}
\begin{lemma}\label{proweiduality}\label{corfirststep}
Let $ C \subset \mathbb{F}_q^n $ be a linear code, $\dim C=k $. Write $ d_r = d_r(C)$, $ d_s^\perp = d_s(C^\perp) $ for $ 1 \leq r \leq k $,  $ 1 \leq s \leq n-k $. Then,
$$ \{1,  \ldots , n\} = \{ d_1,  \ldots , d_k\} \cup \{ n+1-d_{n-k}^\perp, \ldots , n+1-d_1^\perp \}. $$
\end{lemma}
Our first result is a strict improvement to Theorem \ref{theorem tsfasman}.
\begin{theorem} \label{bound 1}
  Let $({\mathcal{F}}_i)_{i=1}^\infty$ be an optimal tower of function fields over ${\mathbb{F}}_q$. Consider $R,\rho$ with $ 1/A(q) \leq R \leq 1 $ and $ \frac{q}{q-1}\frac{1}{A(q)}-\frac{1}{q-1}R \leq \rho \leq R $. Let $(C_i)_{i=1}^\infty$ be an optimal sequence of one-point algebraic geometric codes defined from $({\mathcal{F}}_i)_{i=1}^\infty$ such that $\dim C_i/n_i \rightarrow R$. There exists a sequence of positive integers $(m_i)_{i=1}^\infty$ such that $m_i/n_i \rightarrow \rho$ and  $ d_{m_i}(C_i)/n_i \rightarrow \delta =1-R+\rho$.
\end{theorem}
\begin{proof}
In this proof we use the notation $ k_i = \dim C_i $. Let $ f:{\mathbb{N}} \rightarrow {\mathbb{N}} $ be a function such that $ f(i) \rightarrow \infty $ and $ f(i)/n_i \rightarrow 0 $, as $ i
\rightarrow \infty $. Now fix $ i $. The Goppa bound (Theorem \ref{goppa bound}) together with Lemma \ref{cor2} tell us that
$$ d_{f(i)}(C_i^\perp) \geq \frac{q^{f(i)}-1}{q^{f(i)}-q^{f(i)-1}}(k_i-g(\mathcal{F}_i)). $$
Write $ h(i) $ for the right-hand side, that is, $ d_{f(i)}(C_i^\perp)
\geq \lceil h(i) \rceil $. Observe that $ h(i) > 0 $, since asymptotically $ k_i > g(\mathcal{F}_i) $. If we write $ d_s^\perp = d_s(C_i^\perp) $ for $ 1 \leq s \leq n_i - k_i $, we have that $ n_i + 1 - \lceil h(i) \rceil \geq n_i + 1 - d_{f(i)}^\perp $. From this inequality and the monotonicity of GHWs, it follows that the sets
$$ \{ n_i + 1 - \lceil h(i) \rceil, n_i + 2 - \lceil h(i) \rceil, \ldots, n_i \} \textrm{ and} $$
$$ \{ n_i + 1 - d_{n_i - k_i}^\perp, n_i + 1 - d_{n_i - k_i - 1}^\perp, \ldots, n_i + 1 - d_{f(i) + 1}^\perp \} $$
are disjoint. Therefore, from Lemma \ref{proweiduality} it follows that
\begin{equation}
d_{k_i - \lceil h(i) \rceil + f(i)}(C_i) \geq n_i + 1 - \lceil h(i) \rceil.
\label{almost done}
\end{equation}
Now take a sequence of positive integers $ (m_i)_{i=1}^\infty $ such that 
\begin{equation}
k_i - \lceil h(i) \rceil + f(i) \leq m_i \leq k_i
\label{condition m_i}
\end{equation}
(observe that the left-hand side is smaller than $ k_i $ for large $ i
$). From (\ref{almost done}), (\ref{condition m_i}) and the
monotonicity of GHWs we get
\begin{equation}
\begin{split}
d_{m_i}(C_i) & \geq d_{k_i - \lceil h(i) \rceil + f(i)}(C_i) + m_i - k_i + \lceil h(i) \rceil - f(i) \\
 & \geq n_i-k_i + m_i - f(i) + 1.
\label{ineq ineq}
\end{split}
\end{equation}
Dividing by $ n_i $ and letting $ i \rightarrow \infty $,
(\ref{condition m_i}) and (\ref{ineq ineq}) become
$$ \frac{q}{q-1}\frac{1}{A(q)}-\frac{1}{q-1}R \leq \rho \leq R, $$
$$ \delta = \lim_{i \rightarrow \infty} \frac{d_{m_i}(C_i)}{n_i} = 1-R+\rho. $$
\end{proof}

We have the following result for lower values of $ \rho $.
\begin{theorem} \label{bound 2}
  Let $({\mathcal{F}}_i)_{i=1}^\infty$ be an optimal tower of function fields over ${\mathbb{F}}_q$. Consider $R, \rho$ with $ 0 \leq \rho \leq R \leq 1 $.
  Let $(C_i)_{i=1}^\infty$ be an optimal sequence of one-point algebraic geometric codes defined from $({\mathcal{F}}_i)_{i=1}^\infty$
such that $\dim C_i/n_i \rightarrow R$. For all sequences of positive integers $(m_i)_{i=1}^\infty$ with $m_i/n_i \rightarrow \rho$, the number $ \delta = \liminf_{i \rightarrow \infty} d_{m_i}(C_i) / n_i $ satisfies 
  $$ \delta \geq \frac{q}{q-1} \left( 1-R-\frac{1}{A(q)} \right) +\rho. $$
\end{theorem}
\begin{proof}
Let $ 0 < \varepsilon < 1 $ be an arbitrary fixed number. From the Goppa bound (Theorem \ref{goppa bound}) and Lemma \ref{cor2} we obtain that
$$ \frac{d_{\lceil \varepsilon m_i \rceil}(C_i)}{n_i} \geq
\frac{q^{\varepsilon m_i}-1}{q^{\varepsilon m_i}-q^{\varepsilon m_i
    -1}}\left( 1- \frac{\dim C_i }{n_i} -\frac{g_i}{n_i} \right). $$
Using again the monotonicity of GHWs we obtain that
\begin{equation}
\frac{d_{m_i}(C_i)}{n_i} \geq
\frac{q^{\varepsilon m_i}-1}{q^{\varepsilon
    m_i}-q^{\varepsilon m_i -1}} \left( 1- \frac{\dim C_i}{n_i} -\frac{g_i}{n_i} \right)+\frac{m_i(1-\varepsilon)}{n_i}.\nonumber
\end{equation}
Now, letting $ i \rightarrow \infty $ first and then $ \varepsilon \rightarrow 0 $, we obtain
$$ \delta = \liminf_{i \rightarrow \infty} \frac{d_{m_i}(C_i)}{n_i} \geq \frac{q}{q-1} \left( 1-R-\frac{1}{A(q)} \right) +\rho. $$
\end{proof}

In the following, we concentrate on  Garcia and Stichtenoth's second
tower~\cite{garcia1996asymptotic} of function fields
$({\mathcal{F}}_i)_{i=1}^\infty$ over ${\mathbb{F}}_q$ where $q$ is an
arbitrary perfect square. From~\cite{pellikaan1998weierstrass} we have
a complete description of the corresponding Weierstrass semigroups and
\cite{lowc} gives an efficient method for constructing the corresponding
 optimal sequences of one-point algebraic geometric codes.
We will apply the two new bounds on GHWs given in
Proposition~\ref{proposition1} and Proposition~\ref{proposition2} to
this tower. In the rest of this section, $ q $ is always a perfect square and  by $ (\mathcal{F}_i)_{i=1}^\infty $ we mean Garcia and Stichtenoth's second tower \cite{garcia1996asymptotic}. We will need the following properties of each $ \mathcal{F}_i $ (\cite{garcia1996asymptotic,pellikaan1998weierstrass}): its number of rational places satisfies $ N(\mathcal{F}_i) > q^{\frac{i-1}{2}} (q - \sqrt{q}) $, its genus is given by
$$
g({\mathcal{F}}_i)=\left\{ \begin{array}{ll} (q^{\frac{i}{4}}-1)^2 & {\mbox{ if $i$ is even,}}\\
 (q^{\frac{i+1}{4}}-1)(q^{\frac{i-1}{4}}-1) & {\mbox{ if $i$ is odd,}}\end{array} \right.
$$
and it has a rational place $ Q_i $ such that the conductor of $ H(Q_i) $ is given by
$$
c_i=\left\{ \begin{array}{ll}q^{i/2}-q^{i/4}& {\mbox{ if $i$ is even,}}\\
q^{i/2}-q^{(i+1)/4}& {\mbox{ if $i$ is odd.}}\end{array} \right.
$$
In the rest of the section, $(C_i)_{i=1}^\infty$ is an optimal sequence of one-point algebraic geometric codes defined from $({\mathcal{F}}_i)_{i=1}^\infty$, and where $C_i$ is of the form $C_{\mathcal{L}}(D_i,\mu_iQ_i)$ or $C_{\mathcal{L}}(D_i,\mu_iQ_i)^\perp$. Recall from~\cite{lowc} that we may assume without loss of generality that $D_i$ is chosen in such a way that $C_i$ can be constructed using ${\mathcal{O}}({n_i}^3\log_q^3(n_i))$ operations in ${\mathbb{F}}_q$.
\begin{theorem} \label{bound one-point}
  Let $({\mathcal{F}}_i)_{i=1}^\infty$ be Garcia-Stichtenoth's second tower of function fields over ${\mathbb{F}}_q$, where $q$ is a perfect square. Let $(C_i)_{i=1}^\infty$ be a corresponding optimal sequence of one-point algebraic geometric codes as described above.
  Consider $R,\rho$ with $ 0 \leq R \leq 1 - \frac{1}{\sqrt{q}-1}$ and $ 0\leq \rho \leq \min \{R, \frac{1}{\sqrt{q}-1} \} $, and assume that $\dim C_i /n_i \rightarrow R$. For all sequences of positive integers $(m_i)_{i=1}^\infty$ with $m_i/n_i \rightarrow \rho$, it holds that  $ \delta = \liminf_{i \rightarrow \infty} d_{m_i}(C_i) / n_i $ satisfies 
\begin{equation*}
  \delta \geq 1 - R + 2\rho - \frac{1}{\sqrt{q}-1}.
  \end{equation*}
\end{theorem}
\begin{proof}
We may assume that $ C_i $ is of the form $ C_\mathcal{L}(D_i,\mu_i Q_i) $ or $ C_\mathcal{L}(D_i,\mu_i Q_i)^\perp $, with $ 2g({\mathcal{F}}_i) -2 < \mu_{i} < n_{i} $ and $ (\mu_{i}- g({\mathcal{F}}_i))/n_i \rightarrow R $. As 
$ \lim_{i \rightarrow \infty} c_i/n_i = \lim_{i \rightarrow \infty}
g({\mathcal{F}}_i)/n_i = 1/(\sqrt{q}-1), $ 
the result follows from Proposition \ref{proposition1} or Proposition~\ref{proposition2}.
\end{proof}
\section{The parameters $\Lambda^{(1)}(\varepsilon_1)$ and
  $\Lambda^{(2)}(\varepsilon_2)$ for algebraic geometric code based schemes} \label{sec7}
In Section~\ref{secnewas} we estimated $\Omega^{(1)}$ and
$\Omega^{(2)}$ for asymptotically good sequences of schemes based on algebraic geometric
codes coming from optimal towers of function fields, the sequences
being called asymptotically good if $\Omega^{(1)} > 0$ and
$\Omega^{(2)} < 1$. Employing the
analysis in Section~\ref{sec5} together with~(\ref{eqlambda1bnd2}) and
(\ref{eqlambda2bnd2}) we are now able to give a more complete
picture of the information leakage and reconstruction by providing
also estimates on $\Lambda^{(1)}(\varepsilon_1)$ and
$\Lambda^{(2)}(\varepsilon_2)$. We emphasize that the below theorems apply
also in the cases where one or both of the conditions $\Omega^{(1)}>0$
and $\Omega^{(2)}<1$ fails to hold. Throughout the section recall
that by definition the numbers $\varepsilon_1$ and $\varepsilon_2$ always satisfy $0 \leq
\varepsilon_1, \varepsilon_2 \leq 1$.

\begin{theorem} \label{final1}
For the sequence of linear ramp secret sharing schemes described in
Theorem~\ref{theomega} we have the following estimates: If $1/A(q)
\leq 1-R_2$ and $\varepsilon_1 \geq \big( \frac{q}{q-1}
\frac{1}{A(q)}-\frac{1}{q-1}(1-R_2)\big) / L$ then
$\Lambda^{(1)}(\varepsilon_1) \geq R_2+\varepsilon_1 L$. If $1/A(q) \leq
R_1$ and $\varepsilon_2 \geq \big(
\frac{q}{q-1}\frac{1}{A(q)}-\frac{1}{q-1}R_1 \big) / L$ then
$\Lambda^{(2)}(\varepsilon_2) \leq R_1-\varepsilon_2L$.
\end{theorem}
\begin{proof}
Apply Theorem~\ref{bound 1} with $\rho=\varepsilon_1 L$ and
$\rho=\varepsilon_2 L$, respectively, in combination
with~(\ref{eqlambda1bnd2}) and (\ref{eqlambda2bnd2}), respectively.
\end{proof}
\begin{theorem}\label{final2}
 For the sequence of linear ramp secret sharing schemes described in
Theorem~\ref{theomega} we have the following estimates: 
$\Lambda^{(1)}(\varepsilon_1) \geq \frac{q}{q-1} (R_2 -
\frac{1}{A(q)})+\varepsilon_1L$ and $\Lambda^{(2)}(\varepsilon_2) \leq
\frac{q}{q-1}(R_1+\frac{1}{A(q)})-\frac{1}{q-1}-\varepsilon_2 L$.
\end{theorem}
\begin{proof}
Apply Theorem~\ref{bound 2}  in combination
with~(\ref{eqlambda1bnd2}) and (\ref{eqlambda2bnd2}).
\end{proof}
Observe that from Theorem~\ref{final2} we get an estimate on
$\Lambda^{(1)}(0)$ wich is $q/(q-1)$ times as large as the estimate on
$\Omega^{(1)}$ in Section~\ref{secnewas}. Hence, the studied sequences of
secret sharing schemes are more secure than previously
anticipated. A similar remark holds regarding reconstruction.  
\begin{theorem}\label{final3}
Let $q$ be a perfect square.  For the sequence of linear ramp secret sharing schemes described in
Theorem~\ref{theomega} we have the following estimates:
If $R_2 \geq 1/ (\sqrt{q}-1)$ and $\varepsilon_1 \leq
\frac{1}{\sqrt{q}-1}\frac{1}{L}$ then $\Lambda^{(1)}(\varepsilon_1) \geq
R_2+2\varepsilon_1L-\frac{1}{\sqrt{q}-1}$. If $R_1 \leq
1-\frac{1}{\sqrt{q}-1}$ and $\varepsilon_2 \leq
\frac{1}{\sqrt{q}-1}\frac{1}{L}$ then $\Lambda^{(2)}(\varepsilon_2) \leq
R_1-2 \varepsilon_2 L + \frac{1}{\sqrt{q}-1}$. The $i$-th scheme in the
sequence can be constructed using ${\mathcal{O}}(n_i^3 \log (n_i)^3)$
operations in ${\mathbb{F}}_q$.
\end{theorem}
\begin{proof}
Apply Theorem~\ref{bound one-point}  in combination
with~(\ref{eqlambda1bnd2}) and (\ref{eqlambda2bnd2}).
\end{proof}

We finally remark that when $ q $ is a perfect square, then similarly
to Theorem~\ref{final3}, one can assume in Theorem~\ref{final1} and
Theorem~\ref{final2} that the $ i $-th scheme in the sequence can be
constructed using ${\mathcal{O}}(n_i^3\log (n_i)^3)$ operations in
${\mathbb{F}}_q$.
\section*{Acknowledgments}
The authors gratefully acknowledge the support from The Danish Council
for Independent Research (DFF-4002-00367), from the Spanish 
MINECO/FEDER (MTM2015-65764-C3-2-P), from Japan Society for the Promotion of Science (23246071 and 26289116),
from the Villum Foundation through their VELUX Visiting Professor Programme 2013-2014,
and from the ``Program for Promoting the Enhancement of Research
Universities'' at Tokyo Institute of Technology. They also thank I.\
Cascudo and R.\ Cramer for helpful discussions.


\appendix

\section{Proof of Theorem~\ref{th extended}} \label{app1}

In this appendix we give a proof of Theorem~\ref{th extended}. The
theorem is an improvement of~\cite[Th.\ 9]{aspar},
the improvement stating that the RGHWs of primary and dual nested
linear code pairs can get {\textit{simultaneously}} asymptotically as
close to the Singleton bound (\ref{singleton bound}) as wanted. We use the notation and results in \cite{Goldman,luo,GV,aspar}. In particular, we use the concept of relative dimension length profile (RDLP) as appears in \cite[Sec.\ III]{luo}. For $ 1 \leq d \leq n $, and linear codes $ C_2 \subsetneq C_1 \subset \mathbb{F}_q^n $ define 
\begin{eqnarray} 
K_d(C_1,C_2) &= &\max \{ \dim (C_1
\cap V_I) - \dim(C_2 \cap V_I) \mid \nonumber \\
&&{\mbox{ \hspace{1.8cm} }}I \subset \{ 1, \ldots, n \}, \#
  I = d \},\nonumber
\end{eqnarray}
where $ V_I = \{ \mathbf{x} \in \mathbb{F}_q^n \mid x_i = 0 \textrm{ if } i \notin I \} $. The sequence $ (K_d(C_1,C_2))_{d=1}^n $ is then the RDLP of the pair $ C_2 \varsubsetneq C_1 $ and is known to be non-decreasing \cite[Prop.\ 1]{luo}. Our interest in the RDLP comes from the following result corresponding to the first part of~\cite[Th.\ 3]{luo}:
\begin{equation}
M_m(C_1,C_2)=\min \{ d \mid K_d(C_1,C_2) \geq m\}. \label{eqnyogvigtig}
\end{equation}
As in \cite{Goldman,GV}, we define 
for integers $ a,u,v,w $ the numbers:
$$ N_1(w,u) = \frac{\prod _{i=0}^{u-1} (q^w - q^i)}{\prod _{i=0}^{u-1} (q^u - q^i)}, \quad N_2(w,u,v) = \frac{\prod _{i=0}^{v-1} (q^w - q^{u+i})}{\prod _{i=0}^{v-1} (q^v - q^i)}, $$
and  $N_3(w,u,v,a) = N_1(u,a) N_2(w-a,u-a,v-a)$. 
The meaning of $N_1$ is \cite{Goldman},
\cite[Lem.\ 5 and 6]{GV}:

\begin{lemma}\label{lemitems}
Let $W$ be an ${\mathbb{F}}_q$-linear vector space and let $ u $, $ v
$, $ w = \dim W $ be non-negative integers. If $ u \leq w $, then
  $N_1(w,u)$ is the
  number of subspaces $U \subset W$ of dimension $u$. Furthermore, if $U$
  is fixed and $ u \leq v \leq w $, then $N_1(w-u,v-u)$ is the number of
  ${\mathbb{F}}_q$-linear vector spaces $V$ such that $U\subset V \subset W$
  and $\dim V =v$.
\end{lemma}

From \cite[Lem.\ 9]{GV} we have:
\begin{lemma}\label{ryu1}
Consider fixed integers $ 1 \leq k_2 < k_1 < n $ and a fixed set $I
\subset \{ 1,  \ldots, n\}$ with $\# I=d$. Let $s$ be an integer with
$s \leq \min \{d,k_1-k_2\}$. The number of linear code pairs
$C_2 \varsubsetneq C_1 \subset \mathbb{F}_q^n$
such that $\dim C_1 = k_1$, $\dim C_2=k_2$,
and $\dim (C_1 \cap V_I) - \dim (C_2 \cap V_I) = s $, equals
\begin{eqnarray}
{\mbox{\ \hspace{-0.8cm}}}N_4(n,k_1,k_2,d,s)=\sum_{a=0}^{   \min \{ d-s,
k_1-s,k_2\}  }
\bigg( N_1(d,a) {\mbox{ \hspace{2.3cm}}} 
\nonumber \\
{\mbox{\ \hspace{-0.8cm}}} N_2(n-a,d-a,k_2-a)
N_3 (n-k_2, d-a, k_1-k_2, s)\bigg).\nonumber
\end{eqnarray}
\end{lemma}

We next extend~\cite[Cor.\ 3]{GV}.

\begin{theorem}\label{ryu3}
Consider fixed integers $1 \leq k_2 < k_1 < n$, $ 1 \leq d \leq n $, $ 1 \leq d^\perp \leq n $,
$1 \leq s \leq \min\{d$, $k_1-k_2\}$, and $1 \leq s^\perp \leq 
\min\{d^\perp$, 
$k_1-k_2\}$. There exists a nested linear code pair
$C_2 \subsetneq C_1 \subset \mathbb{F}_q^n$
such that 
$\dim C_1 = k_1$, $\dim C_2=k_2$,
$M_s(C_1, C_2) > d$ and
$M_{s^\perp}(C_2^\perp, C_1^\perp) > d^\perp$,
if
\begin{eqnarray}
{\mbox{ \ \hspace{-0.6cm}}}N_1(n,k_2) N_1(n-k_2, k_1-k_2) > {n \choose d}
\sum_{\sigma=s}^{k_1-k_2}  N_4(n,k_1,k_2,d,\sigma) {\mbox{
  \hspace{0.2cm}}} \nonumber \\
+ {n \choose d^\perp}
\sum_{\sigma^\perp=s^\perp}^{k_1-k_2}N_4(n,n-k_2,n-k_1,d^\perp ,\sigma^\perp).\nonumber
\end{eqnarray}
\end{theorem}
\begin{proof}
By Lemma~\ref{lemitems}, the term $N_1(n,k_2) N_1(n-k_2, k_1-k_2)$
is the total number of pairs
$C_2 \subsetneq C_1 \subset \mathbb{F}_q^n$
such that $\dim C_1 = k_1$ and $\dim C_2=k_2$. On the other hand, by Lemma~\ref{ryu1}, the number of pairs
$C_2 \subsetneq C_1 \subset \mathbb{F}_q^n$
such that $\dim C_1 = k_1$, $\dim C_2=k_2$ and $K_d(C_1,C_2) \geq
s$ is at most 
${n \choose d}
\sum_{\sigma=s}^{k_1-k_2}  N_4(n,k_1,k_2,d,\sigma)$. Similarly, the number of pairs
$C_2 \subsetneq C_1 \subset \mathbb{F}_q^n$
such that $\dim C_1 = k_1$, $\dim C_2=k_2$ and
$K_{d^\perp}(C_2^\perp ,C_1^\perp )\geq s^\perp$ is at most 
${n \choose d^\perp}
\sum_{\sigma^\perp=s^\perp}^{k_1-k_2}N_4(n,n-k_2,n-k_1,d^\perp
,\sigma^\perp)$. The inequality therefore ensures the
existence of a code pair $C_2 \subsetneq C_1 \subset
\mathbb{F}_q^n$ with $\dim C_1 =k_1$, $\dim C_2 = k_2$,  $K_{d}(C_1,C_2) < s$ and $K_{d^\perp}(C_2^\perp, C_1^\perp) < s^\perp$. But the RDLP is non-decreasing
  and $K_{n}(C_1,C_2)=K_n(C_2^\perp,C_1^\perp)=k_1-k_2$ which is
    larger than or equal to $s$ and $s^\prime$. Therefore there exists
    a smallest index $j$ such that $K_j(C_1,C_2) \geq s$ and a
    smallest index $j^\perp$ such that $K_{j^\perp}(C_2^\perp ,
    C_1^\perp ) \geq s^\perp$ and $j > d$ as well as $j^\perp >
    d^\perp$ hold. The theorem now follows
from~(\ref{eqnyogvigtig}). \end{proof} 

To apply Theorem~\ref{ryu3} in an asymptotic setting we will need a
couple of lemmas.
\begin{lemma}\label{lemA}
Define $\pi (q)=\prod_{i=1}^\infty (1-q^{-i})$. Then
\begin{eqnarray}
&\pi(q) q^{u(w-u)} \leq N_1(w,u) \leq
  \pi(q)^{-1}q^{u(w-u)}, \label{eqell6}\\
&N_2(w,u,v) \leq \pi (q)^{-1} q^{v(w-v)}, \nonumber \\
&N_3(w,u,v,a) \leq \pi (q)^{-2} q^{a(u-a)}q^{(v-a)(w-v)}. \label{eqell8}
\end{eqnarray}
\end{lemma}
\begin{proof}
The inequality~(\ref{eqell6}) is \cite[Cor.\ 2]{kilde5} and the last
two inequalities correspond to~\cite[Lem.\ 3]{aspar} except that 
$\pi(q)^{-2}$ in~(\ref{eqell8}) by a mistake was there
written $\pi(q)^{-1}$ and similarly $q^{a(u-a)}$ was written $q^{u(u-a)}$. 
\end{proof}
The next lemma corresponds to~\cite[Ex.\ 11.1.3]{cover}.
\begin{lemma}\label{lemB}
\hspace{-0.1cm}Let $H_q(x)=-x \log_q (x)-(1-x)\log_q(1-x)$, then 
$$\frac{1}{n+1} q^{n H_q(m/n)} \leq {n \choose m} \leq q^{nH_q(m/n)}.$$
\end{lemma}
With the above machinery we can now give the promised proof.\\

\noindent {\bf{Proof of Theorem~\ref{th extended}. \ }}
Let $R_1$, $R_2$, $\delta$, $\delta^\perp$, $\tau$ and $\tau^\perp$ be
as in the theorem (in particular assume 
(\ref{eqnummer8}) to hold). Let $(n_i)_{i=1}^\infty$ be a strictly
increasing sequence of positive integers and define $k_1(i)=\lfloor n_i R_1\rfloor$, 
$k_2(i)=\lceil n_i R_2\rceil$, $s(i)=\lceil  n_i \tau \rceil$,
$s^\perp(i)=\lceil n_i \tau^\perp  \rceil$, $d(i)=\lfloor  n_i \delta
\rfloor$ and $d^\perp(i)=\lfloor n_i \delta^\perp 
\rfloor$. Using Theorem~\ref{ryu3}, we will show that for $i$ large
enough there exist nested linear codes $C_2(i) \subsetneq C_1(i) \subset
{\mathbb{F}}_q^{n_i}$ of dimensions $k_2(i)$ and $k_1(i)$, respectively, with
\begin{equation}
M_{s(i)} \geq d(i), \quad \textrm{and} \quad M_{s^\perp(i)} \geq d^\perp (i). \label{eqsnabel1plus2} 
\end{equation}
Observe that (\ref{eqnummer8}) implies that
\begin{equation}
k_1(i) +d(i)-n_i-s(i) < 0, \label{stjerne1}
\end{equation}
\begin{equation}
(n_i-k_2(i))+d^\perp (i) -n_i-s^\perp (i)<0, \label{stjerne2}
\end{equation}
which we will need later in the proof. For brevity, we will write $k_1$, $k_2$, $d$, $d^\perp$, $s$,
$s^\perp$, and $n$ rather than $k_1(i)$, $k_2(i)$, $d(i)$, $d^\perp
(i)$, $s(i)$, $s^\perp (i)$, and $n_i$. Applying Lemma~\ref{lemA},
Lemma~\ref{lemB}, and Theorem~\ref{ryu3} we see that a sufficient condition for the existence of a linear code pair satisfying (\ref{eqsnabel1plus2}) is
\begin{eqnarray}
{\mbox{ \hspace{-0.8cm}}}&&\pi (q)^2 q^{k_2(n-k_2)}q^{(k_1-k_2)(n-k_1)}\nonumber \\
{\mbox{ \hspace{-0.8cm}}}&>&q^{nH_q(d/n)}\sum_{\sigma =s}^{k_1-k_2} \sum_{a=0}^{\min \{
    d-\sigma, k_1-\sigma, k_2\}} \bigg[\pi(q)^{-1} q^{a(d-a)}
    \nonumber \\
{\mbox{ \hspace{-0.8cm}}}&&\pi(q)^{-1}
   q^{(k_2-a)(n-a-k_2+a)}\pi(q)^{-2}q^{\sigma(d-a-\sigma)}q^{(k_1-k_2-\sigma)(n-k_2-k_1+k_2)}
   \bigg] \nonumber \\
{\mbox{ \hspace{-0.8cm}}}&&+q^{nH_q(d^\perp /n)}\sum_{\sigma^\perp = s^\perp}^{k_1-k_2}
   \sum_{a=0}^{\min \{d^\perp -\sigma^\perp, n-k_2-\sigma^\perp,
   n-k_1\}}\bigg[ \pi(q)^{-1} q^{a(d^\perp -a)} \nonumber \\
{\mbox{ \hspace{-0.8cm}}}
&&\pi (q)^{-1}q^{(n-k_1-a)(n-a-n+k_1+a)}\pi(q)^{-2}q^{\sigma^\perp (d^\perp
   -a -\sigma^\perp)}q^{(k_1-k_2-\sigma^\perp)k_2}\bigg]. \nonumber
\end{eqnarray}
But then another sufficient condition (named
Condition A) for the existence of a nested code pair satisfying
(\ref{eqsnabel1plus2}) is
\begin{eqnarray}
{\mbox{\hspace{-0.8cm}}}&&q^{k_2(n-k_2)+(k_1-k_2)(n-k_1)} >\nonumber \\
{\mbox{\hspace{-0.8cm}}}&& f(q,n) \max \bigg\{ q^{a(d-a)+(k_2-a)(n-k_2)+\sigma
    (d-a-\sigma)+(k_1-k_2-\sigma)(n-k_1)} \bigm| \nonumber \\
{\mbox{\hspace{-0.8cm}}}&& {\mbox{ \hspace{0.9cm} }} s \leq \sigma  \leq k_1-k_2, 0 \leq a \leq \min
\{d-\sigma, k_1-\sigma,k_2 \} \bigg\} + \nonumber \\
{\mbox{\hspace{-0.8cm}}}&&f^\perp (q,n) \max \bigg\{
   q^{a(d^\perp-a)+(n-k_1-a)k_1+\sigma^\perp (d^\perp - a -
   \sigma^\perp)+(k_1-k_2-\sigma^\perp)k_2} \bigm|  \nonumber \\
{\mbox{\hspace{-0.8cm}}}&&{\mbox{ \hspace{1.8cm} }} s^\perp \leq \sigma^\perp
\leq k_1-k_2, {\mbox{ and }}\nonumber \\
{\mbox{\hspace{-0.8cm}}}&& {\mbox{\hspace{2cm}}} 0 \leq a \leq \min
\{d^\perp-\sigma^\perp, n-k_2-\sigma^\perp,n-k_1\} \bigg\}, \nonumber
\end{eqnarray}
where $f(q,n)=\pi (q)^{-6}q^{nH_q(d/n)}n^2$, and  where $f^\perp(q,n)=\pi (q)^{-6}q^{nH_q(d^\perp /n)}n^2$. 
Consider now the expression $\sigma(k_1+d-n-\sigma-a)$, which contains
the terms in the first exponent on the right-hand side of Condition
A related to $\sigma$. As a function in $\sigma$, this is a downward
parabola intersecting the first axis in $\sigma=0$. For $s \leq
\sigma$, it follows from~(\ref{stjerne1}) and $0 \leq a$ that
$k_1+d-n-\sigma - a <0$. Hence, the maximal value of $\sigma
(k_1+d-n-\sigma-a)$ for $s \leq \sigma$ is attained when $\sigma =s$,
and we therefore substitute $\sigma$ with $s$ in Condition A. In a
similar fashion, we see from~(\ref{stjerne2}) that $\sigma^\perp$ can
be replaced with $s^\perp$. After these substitutions, the terms
related to $a$ in the first exponent on the right-hand side of
Condition A become $-a^2+a(k_2+d-n-s)$, which is equal to $0$ for $a=0$
and negative for $a >0$, as a consequence
of~(\ref{stjerne1}). Similarly, the terms related to $a$ in the last
exponent on the right-hand side become $-a^2+a(d^\perp - k_1-s^\perp)$
which again is equal to $0$ for $a=0$ and negative for $a >0$ as a
consequence of~(\ref{stjerne2}). Hence, we can substitute $a$ with $0$
in Condition A. After the above substitutions, Condition A simplifies to
\begin{eqnarray}
{\mbox{\hspace{-0.6cm}}} q^{k_2(n-k_2)+(k_1-k_2)(n-k_1)} &&> f(q,n) q^{k_2(n-k_2)+s
    (d-s)+(k_1-k_2-s)(n-k_1)} \nonumber \\
&&{\mbox{ {\hspace{-0.2cm} }}}+f^\perp (q,n)
   q^{(n-k_1)k_1+s^\perp (d^\perp - 
   s^\perp)+(k_1-k_2-s^\perp)k_2}. \nonumber 
\end{eqnarray}
In this formula, we now replace the two
expressions on the right-hand side with the largest one multiplied
by 2. We then take the logarithm over $ q $ and finally divide by $n^2$. Assume that the first term on the right-hand side of Condition A is greater than or equal to the last term. After simplifying equal terms on both sides and using the definition of $ k_1 $, $ d $ and $ s $, we see that Condition A holds if
\begin{equation}
0 > g(i)+\tau(\delta - \tau)-\tau(1-R_1), \label{eqff1}
\end{equation}
where $ g(i) = \log_q(2f(q,n_i)) / n_i^2 $, which goes to $ 0 $ as $ i $ goes to infinity. Similarly, if the last term on the right-hand
side is greater than or equal to the first term, we see that Condition A holds if
\begin{equation}
0 > g^\perp(i) + \tau^\perp (\delta^\perp -\tau^\perp) - \tau^\perp R_2, \label{eqff2}
\end{equation}
where $ g^\perp(i) = \log_q(2f^\perp(q,n_i)) / n_i^2 $, which again
goes to $ 0 $ as $ i $ goes to infinity. Finally, for $ i $ large enough, (\ref{eqff1}) follows from the first
part of (\ref{eqnummer8}), since $ \tau > 0 $, and (\ref{eqff2})
follows from the last part of (\ref{eqnummer8}), since $ \tau^\perp >
0 $. Therefore, Condition A holds for $ i $ large enough and we are
done. {\hfill{$\Box$}}


\begin{thebibliography}{25}

\bibitem{bassabeelenmv}
{A.~Bassa}, {P.~Beelen},
  {A.~Garcia}, {H.~Stichtenoth},
  {Towers of function fields over non-prime finite fields},
  {\textit{{Moscow Mathematical Journal}}}, {15}:{1--29}, {2015}. 

\bibitem{blakleyoprindelig}
{G.R.~Blakley}, {Safeguarding
  cryptographic keys}, {{\textit{Proc. of the National
  Computer Conference 1979}}},
{48}:{313--317}, {1979}



\bibitem{blakley}    
{G.R.~Blakley}, 
{C.~Meadows}, {Security of ramp
  schemes}, {Advances in cryptology---{CRYPTO} 1984}, 
\textit{{Lecture Notes in Comput.
  Sci.}}, {196}:{242--268}, {1995}. 
\bibitem{c2}
{I.~Cascudo}, {R.~Cramer},
  {C.~Xing}, {Bounds on the threshold gap in
  secret sharing and its applications}, {\textit{{IEEE Trans. Inform.
  Theory}}}, {59}:{5600--5612}, {2013}. 
\bibitem{MR3225937}
{I.~Cascudo}, {R.~Cramer},
  {C.~Xing}, {Torsion limits and
  {R}iemann-{R}och systems for function fields and applications},
  {\textit{{IEEE Trans. Inform. Theory}}}, {60}:{3871--3888}, {2014}. 
\bibitem{MR2422182}
{H.~Chen}, {R.~Cramer},
  {Algebraic geometric secret sharing schemes and secure
  multi-party computations over small fields}, in: {Advances
  in cryptology---{CRYPTO} 2006}, 
  \textit{{Lecture Notes in Comput. Sci.}}, {4117}:{521--536}, {2006}. 
\bibitem{Chen}
{H.~Chen}, {R.~Cramer},
  {S.~Goldwasser}, {R.~de~Haan},
  {V.~Vaikuntanathan}, {Secure computation from
  random error correcting codes}, in: {Advances in
  cryptology---{EUROCRYPT} 2007}, 
  \textit{{Lecture Notes in Comput. Sci.}}, {4515}:{291--310}, {2007}.
\bibitem{Cramer-Cascudo}
{H.~Chen}, {R.~Cramer},
  {R.~de~Haan}, {I.~Cascudo},
  {Strongly multiplicative ramp schemes from high degree
  rational points on curves}, in: {Advances in
  cryptology---{EUROCRYPT} 2008}, 
  \textit{{Lecture Notes in Comput. Sci.}}, {4965}:{451--470}, {2008}. 
\bibitem{cover}
  {T.~M.~Cover}, {J.~A.~Thomas}, 
  {\textit{{Elements of Information
    Theory}}}, {2nd edition}, {Wiley
    Interscience}, {2006}.

\bibitem{cramer3}
{R.~Cramer}, {I.B. Damg{\aa}rd},
  {N.~D\"{o}ttling}, {S.~Fehr},
  {G.~Spini}, {Linear secret sharing schemes from
  error correcting codes and universal hash functions}, in: {Advances in
  cryptology---{EUROCRYPT} 2015}, 
  \textit{{Lecture Notes in Comput. Sci.}}, {9057}:{313--336}, {2015}. 
\bibitem{garcia1996asymptotic}
{A.~Garcia}, {H.~Stichtenoth},
  {On the asymptotic behaviour of some towers of function fields
  over finite fields}, {\textit{{Journal of Number Theory}}},
  {61}:{248--273}, {1996}.
\bibitem{rghw2014}
{O.~Geil}, {S.~Martin},
  {R.~Matsumoto}, {D.~Ruano},
  {Y.~Luo}, {Relative generalized {H}amming
  weights of one-point algebraic geometric codes}, {\textit{{IEEE
  Trans. Inform. Theory}}}, {60}:{5938--5949}, {2014}. 
\bibitem{Goldman}
{J.~Goldman}, {G.-C.~Rota},
{On the foundations of combinatorial theory IV: Finite
  vector spaces and Eulerian generating functions},
{\textit{{Studies in Applied
  Mathematics}}}, {49}:{239--258}, {1970}.

\bibitem{kilde5}
{T.~Helleseth}, {T.~Kl\"{o}ve},
{V.~I.~Leveshtein}, {{\O}.~Ytrehus},
{Bounds on the minimum support weights}, 
 {\textit{{IEEE Trans. Inform. Theory}}}, {41}:{432--440}, {1995}. 



\bibitem{handbook}
{T.~H{\o}holdt}, {J.H. van Lint},
  {R.~Pellikaan}, {Algebraic geometry codes},
  in: {V.S. Pless}, {W.C. Huffman} (Eds.),
  {\textit{{Handbook of Coding Theory}}}, {1}:{871--961}, 
  {Elsevier}, {Amsterdam},
 {1998}.

\bibitem{ihara}
{Y.~Ihara},
  {Some remarks on the number of rational points of algebraic curves over finite fields}, {\textit{{J. Fac. Sci. Tokyo}}},
  {28}:{721-–724}, {1981}.


\bibitem{kurihara}
{J.~Kurihara}, {T.~Uyematsu},
  {R.~Matsumoto}, {Secret sharing schemes based
  on linear codes can be precisely characterized by the relative generalized
  {H}amming weight}, {\textit{{IEICE
      Trans. Fundamentals}}}, {E95-A}:{2067--2075}, {2012}.
\bibitem{luo}
{Y.~Luo}, {C.~Mitrpant},
  {A.J. Han Vinck}, {K.~Chen},
  {Some new characters on the wire-tap channel of type {II}},
  {\textit{{IEEE Trans. Inform. Theory}}}, {51}:{1222--1229}, {2005}. 
\bibitem{GV}
{R.~Matsumoto}, {Gilbert-Varshamov-type bound
  for relative dimension length profile}, {\textit{{IEICE Comm.
  Express}}}, {2 (8)}:{343--346}, {2013}. 
\bibitem{aspar}
{R.~Matsumoto}, {New asymptotic metrics for
  relative generalized {H}amming weight}, {\textit{{Proceedings of IEEE
  International Symposium on Information Theory}}},   
  {3142--3144}, {2014}.
\bibitem{pellikaan1998weierstrass}
{R.~Pellikaan}, {H.~Stichtenoth},
  {F.~Torres}, {{W}eierstrass semigroups in an
  asymptotically good tower of function fields},
{\textit{{Finite Fields Appl.}}}, {4}: 
  {381--392}, {1998}. 
\bibitem{shamir}
{A. Shamir}, {How to share a secret},
{\textit{{Commun. ACM}}}, {22} ({11}):{612--613}, {1979}.
\bibitem{lowc}
{K.W. Shum}, {I.~Aleshnikov},
  {P.V. Kumar}, {H.~Stichtenoth},
  {V.~Deolaikar}, {A low-complexity algorithm
  for the construction of algebraic-geometric codes better than the
  {G}ilbert-{V}arshamov bound}, {\textit{{IEEE
      Trans. Inform. Theory}}}, {47}:{2225--2241}, {2001}.
\bibitem{tsfasman95}
{M.A. Tsfasman}, {S.G. Vl{\u{a}}du{\c{t}}},
  {Geometric approach to higher weights}, {\textit{{IEEE
  Trans. Inform. Theory}}}, {41}:{1564--1588}, {1995}. 
\bibitem{drvl}
{S.G. Vl{\u{a}}du{\c{t}}}, {V.G.
  Drinfeld}, {The number of points of an algebraic
  curve}, {\textit{{Funktsional. Anal. i Prilozhen.}}},
  {17}:{68--69}, {1983}.
\bibitem{wei}
{V.K. Wei}, {Generalized {H}amming weights for
  linear codes}, {\textit{{IEEE Trans. Inform. Theory}}},
  {37}:{1412--1418}, {1991}.
\bibitem{yamamoto}
{H.\ Yamamoto}, {Secret sharing system using (k,L,n) threshold scheme}, {\textit{{Electronics and
      Communications in Japan (Part I: Communication)}}},
{69}:{46--54}, {1986}


\end{thebibliography}
\end{document}